\documentclass[conference,a4paper]{IEEEtran}

\usepackage[utf8]{inputenc}

\usepackage{amsmath, mathrsfs, mathtools}
\usepackage{amsfonts}
\usepackage{epstopdf}
\usepackage{amssymb}
\usepackage{dsfont}
\usepackage{color}
\usepackage{graphics}
\graphicspath{{../}{./}}
\usepackage{algorithm}
\usepackage{algpseudocode}
\usepackage{acro}
\usepackage[belowskip=-10pt,aboveskip=5pt,font=footnotesize]{caption}

\usepackage{tikz}
\usetikzlibrary{shapes.geometric, arrows, positioning, fit, patterns}

\usepackage[backend=bibtex,sorting=none,style=ieee]{biblatex}
\usepackage{balance}
\renewcommand*{\bibfont}{\footnotesize}

\let\labelindent\relax
\usepackage{enumitem}

\newcommand{\subparagraph}{}
\usepackage[compact]{titlesec}
\titlespacing*{\section}{15pt}{1.2\baselineskip}{0.9\baselineskip}

\newtheorem{theorem}{Theorem}
\newtheorem{corollary}{Corollary}
\newtheorem{remark}{Remark}
\newtheorem{lemma}{Lemma}

\long\def\comment#1{}

\newcommand\figref{Fig.~\ref}

\newcommand{\ben}{\begin{enumerate}}
\newcommand{\een}{\end{enumerate}}

\newcommand{\beq}{\begin{equation}}
\newcommand{\eeq}{\end{equation}}

\newcommand{\bi}{\begin{itemize}}
\newcommand{\ei}{\end{itemize}}

\newcommand{\PP}{\mathbb{P}}

\newcommand{\EE}{\mathbb{E}}

\newcommand{\bv}{{\bf b}}
\newcommand{\cv}{{\bf c}}
\newcommand{\dv}{{\bf d}}
\newcommand{\ev}{{\bf e}}

\newcommand{\gv}{{\bf g}}
\newcommand{\hv}{{\bf h}}

\newcommand{\mv}{{\bf m}}

\newcommand{\sv}{{\bf s}}

\newcommand{\uv}{{\bf u}}

\newcommand{\vv}{{\bf v}}

\newcommand{\yv}{{\bf y}}

\newcommand{\zerov}{{\bf 0}}
\newcommand{\onev}{{\bf 1}}

\newcommand{\Gm}{{\bf G}}
\newcommand{\Hm}{{\bf H}}

\newcommand{\lambdav}{\hbox{\boldmath$\lambda$}}

\IEEEoverridecommandlockouts
\title{On the Advantages of Asynchrony in the Unsourced MAC}
\author{Alexander Fengler, %
Alejandro Lancho, %
Krishna Narayanan,
and Yury Polyanskiy%
\IEEEcompsocitemizethanks{
\IEEEcompsocthanksitem A. Fengler, A. Lancho, and Y. Polyanskiy are with the Massachusetts Institute of Technology.
(Email: \{fengler,lancho,yp\}@mit.edu)
\IEEEcompsocthanksitem Research was sponsored by the United States Air Force Research Laboratory and the United States Air Force Artificial Intelligence Accelerator and was accomplished under Cooperative Agreement Number FA8750-19-2-1000. The views and conclusions contained in this document are those of the authors and should not be interpreted as representing the official policies, either expressed or implied, of the United States Air Force or the U.S. Government. The U.S. Government is authorized to reproduce and distribute reprints for Government purposes notwithstanding any copyright notation herein. Alejandro Lancho has received funding from the European Union’s Horizon 2020 research and innovation programme under the Marie Sklodowska-Curie grant agreement No 101024432. Alexander Fengler was funded by the Deutsche Forschungsgemeinschaft (DFG, German Research Foundation) – Grant 471512611. This work is also supported by the National Science Foundation (NSF) under Grant No CCF-2131115.
\IEEEcompsocthanksitem K. Narayanan is with the Texas A\&M University. This work is supported by NSF grant CCF-2131106. (Email: krn@tamu.edu)}%
}

\begin{document}
\maketitle
\begin{abstract}

In this work we demonstrate how a lack of synchronization can in fact be advantageous in the problem of random access. Specifically, 
we consider a multiple-access problem over a frame-asynchronous 2-user binary-input adder channel in the unsourced setup (2-UBAC). Previous work has shown that under perfect synchronization the per-user rates achievable with linear codes over the 2-UBAC are limited by 0.5 bit per channel use (compared to the capacity of 0.75). In this paper, we first demonstrate that arbitrary small (even single-bit) shift between the user's frames enables (random) linear codes to attain full capacity of 0.75 bit/user. Furthermore, we derive density evolution equations for irregular LDPC codes, and prove (via concentration arguments) that they correctly track the asymptotic bit-error rate of a BP decoder. Optimizing the degree distributions we construct LDPC codes achieving per-user rates of 0.73 bit per channel use.    
\end{abstract}
\begin{keywords}
Multiple-Access, Low-density parity check (LDPC), Unsourced, massive machine-type communication   	
\end{keywords}

\section{Introduction}

A recent line of work, termed unsourced random access (URA or UMAC), exploits the idea of same-codebook communication \cite{Pol2017}. This approach allows to separate the different messages in a multiple-access channel (MAC) based purely on the structure of the codebook, i.e., the set of allowed messages. It was shown that good unsourced code designs can approach the capacity of the additive white Gaussian noise (AWGN) adder channel without the need for coordination \cite{Pol2017,Fen2021f}. While many unsourced code constructions have been proposed \cite{Fen2021f,Pra2020a,Mar2019,Kow2020,Ama2020a,Tru2021a,Fen2022}, most of them lack analytic understanding and it is not well understood what properties make a good unsourced codebook.
Furthermore, many proposed schemes have a high decoding complexity. 
Recent works \cite{Liv2021a,Fen2022a} have constructed LDPC codes specifically for two-user communication on the unsourced binary input adder channel (UBAC). It was found that linear codes in general suffer a rate loss in the UBAC and cannot achieve sum rates higher than 1 bit/channel use, which is still far from the sum-rate capacity of 1.5 bits/channel use. 

Another concern for the practical applicability of unsourced codes is the assumption of perfect synchronization, present in many works. In low-power low-cost transmitters perfect synchronization is hard to achieve. Classic results \cite{Cov1981} show that frame-asynchrony does not change the capacity of a discrete MAC, as long as the allowed delay is smaller than the blocklength. 
Recent solutions for uncoordinated multiple-access schemes that can deal with asynchronism were proposed in \cite{Dec2022,Che2022a}. 
Both of these works present schemes specifically for orthogonal frequency-division multiplexing (OFDM) modulation with timing offsets within the cyclic prefix. Such timing offsets can be efficiently handled in the frequency domain. Nonetheless, OFDM is not necessarily the best choice for the mMTC scenario since it requires a high level of frequency synchronization, which is hard to achieve with low-cost transmitters. 

In this work, we first show that random linear codes achieve the BAC capacity of 1.5 bits/ch. use as soon as a frame delay of at least one symbol is introduced. As such, it enables same-codebook communication with linear codes and linear decoding complexity that does not suffer from the rate 1 bottleneck, which limits unsourced linear codes in the frame-synchronous case. 
Although the channel model is idealistic, it is also quite general and does not rely on any specific modulation method.
Further, we design LDPC codes with linear decoding complexity for the two-user frame-asynchronous UBAC.  We find codes that
achieve sum-rates of 1.46 bits/ch. use. The decoding can be done by two copies of a conventional single-user belief propagation (BP) decoder that periodically exchange information. We also show that our design works if the delay is a random integer with a maximum value that scales at most sub-linearly with the blocklength. 

Randomized LDPC code designs for the two-user multiple-access channel with AWGN have been
presented in \cite{Rou2007,Bal2019}.  For the code construction presented in \cite{Bal2019} it is
crucial that the two code ensembles are optimized independently, resulting in two different
ensembles. If one check node (CN) distribution is fixed, the CN distribution of the other user can
be optimized by a linear program.  In \cite{Rou2007}, one common code ensemble is designed, but
the two users pick a different random code from the same ensemble. In addition, to obtain a linear
optimization program, the codes in \cite{Rou2007} are constrained such that variable nodes (VNs)
that are connected through the MAC have the same degree.  Such a constraint would be hard to
enforce in a model with random delay.  In contrast, in this work we design one LDPC ensemble from
which one code is chosen at random and used by both users. The design of the ensemble relies on
alternating optimization of CN and VN degree distributions.  Surprisingly, we find that degree one
VNs do not result in error floors, in contrast to LDPC codes for the single-user binary-erasure
channel (BEC). A particular difficulty in proving the density evolution (DE) in the joint graph
is that the channel transition probabilities for one user depends on the transmitted codeword of the other user.
Since the codewords come from the same codebook the channel outputs may be correlated.
To that end we employ the
symmetrization technique of coset ensembles, cf. \cite{Gal1968}, although
an additional subtlety in our case is that we need to
show that both users can use the same coset. Thus, our design
strictly adheres to the unsourced paradigm where both users use a common codebook. 
The symmetrization allows us
to prove that DE describes the asymptotic bit-error rate (BER) and, furthermore, that it is
independent of the transmitted codewords. This implies that we can assume that both users transmit
the all-zero codeword plus a dither when analyzing the error probability.  We provide a full
proof that the asymptotic error probability is described by the DE and give an analysis of the
probability of short-length stopping sets, which result in an error floor.  The error floor analysis
shows that we can expurgate short-length stopping sets created by the MAC nodes as long as the fraction
of degree one VNs is below a certain threshold.  Numerical simulations confirm that DE accurately
predicts the error probability for large blocklengths. We use the DE to construct codes that
approach the capacity of the two-user BAC.     Our work shows that frame-asynchrony can be
exploited to design efficient linear unsourced codes.

To summarize, our main intellectual contributions in this paper are:
\begin{itemize}[leftmargin=*]
    \item A random coding argument that shows that linear codes can achieve the full BAC capacity with a single symbol delay.
    \item The derivation of the DE equations under the same-codebook constraint and sub-linear frame delays.
    \item A rigorous proof that the BER of a random code from the ensemble will concentrate around the DE.
    \item The design of a codebook that enables two-user communication at rates close to the Shannon limit.
\end{itemize}
These findings imply that a non-zero frame delay enables two users to use the same LDPC encoder while still achieving rates close to the two-user BAC capacity. In addition, decoding can be done with linear complexity and a simplified decoder architecture that consists of two connected copies of the same single-user BP decoder.

\section{Channel Model}
\label{sec:channel}
We study the frame-asynchronous noiseless BAC:
\begin{equation}
	y_i = c_{1,i} + c_{2,i-\tau}
	\label{eq:channel}
\end{equation}
where $\tau \in [0:\tau_\text{max}]$ and $c_{u,i} \in \{1,-1\}$ for $u \in \{1,2\}, i\in [1:n]$ and 
$c_{u,i} = 0$ for $i < 1$ or $i>n$.
More specifically, each user transmits a binary-phase-shift keying (BPSK) modulated version of a binary codeword $\cv_{u} = 2\mv_u - 1, \mv_u \in \{0,1\}^n$. We will analyze the case where $\tau$ is random and uniformly distributed. Furthermore, we will study the asymptotic behavior of code constructions when $\tau_\text{max}\in o(n)$, i.e., $\tau_\text{max}/n \to 0$ as $n\to\infty$. This setting is also known as mild asynchrony in information theory \cite{Gam2011}. Both users
transmit a uniform i.i.d. sequence of $nR$ bits, $\bv_1, \bv_2$, by picking the respective binary codewords $\mv_1,\mv_2$ independently, uniformly at random from a common codebook 
 over the binary field $\mathcal{C} \in \mathbb{F}_2^{n\times 2^{nR}}$, where $n$ denotes the blocklength and $0 < R < 1$ the per-user rate. The decoder outputs a list of two messages $g(\yv)$ and the per-user error probability is defined as
 $P_e = \frac{1}{2}(\PP(\bv_1 \notin g(\yv)) + \PP(\bv_2 \notin g(\yv)))$.

Since the model includes no noise, the channel model reduces to an erasure channel where a received symbol can be considered as erased if 
$(c_{1,i},c_{2,i-\tau}) \in \{(+1,-1),(-1,+1)\}$. 
\begin{remark}
	The coding construction in this paper also works for the synchronous model if users employ a randomly chosen \emph{cyclic} shift of their codeword before transmission. However, in this case some mechanism needs to be added that allows to recover the shift of each user, e.g., adding a preamble to each codeword. For the model \eqref{eq:channel} this is not necessary since $\tau$ can be found easily from amplitude information in $\yv$.
\end{remark}
\begin{remark}
The BAC model can also be used to model on-off keying modulation. In that case there is some ambiguity since there is no dedicated idle symbol. Nonetheless, it is still possible to detect the start of a frame by introducing a preamble. 
\end{remark}
\section{Random Linear Codes}
We give the following result, which shows that random linear codes can achieve the two-user BAC capacity if a frame delay of just one symbol
is introduced.
\begin{theorem}
    \label{thm:rlc}
    There exist linear $(n,k)$ codes for the two-user frame-asynchronous UBAC with $\tau = 1$ and
    \begin{equation}
        P_e \leq \frac{n-1}{2}2^{n(2R - 1.5)} + o_n(1).
    \end{equation}
     \hfill $\square$
\end{theorem}
\begin{proof}
   The proof is given in Appendix \ref{appendix:rlc}.
\end{proof}
Theorem \ref{thm:rlc} shows that random linear codes can achieve a vanishing error probability if $R < 0.75 - \delta$ for any $\delta>0$. It can be shown for both parity check and generator ensembles.
We briefly describe the intuition behind the proof for parity check ensembles and why $\tau>0$ is strictly necessary to get rates larger than $0.5$. The idea is to treat the channel as erasure channel, as described in Section \ref{sec:channel}. The erased symbols can, in principle, be recovered by solving the parity check equations $\Hm \mv_1 = \zerov$ and $\Hm \mv_2 = \zerov$. A key property of the BAC is that on the erased set the codewords from the two user have opposed bits, i.e. $c_{1,i} = - c_{2,i-\tau}$. This gives a second collection of parity equations for each codeword. For $\tau=0$ the additional parity check equations would be linearly dependent, and provide no new information. In that case, since the size of the erased set is around $n/2$, the parity check matrix needs to have $n/2 + \delta$ linearly independent rows for correct recovery, resulting in $R<1/2$. In contrast, for $\tau=1$ we show in Appendix \ref{appendix:rlc} that the collection of parity check equations arising from $c_{1,i} = - c_{2,i-\tau}$ for $i\in \mathcal{E}$ is linearly independent from the set of equations given by $\Hm \mv_1 = \Hm \mv_2 = \zerov$ with high probability. Therefore $n/4 + \delta$ linearly independent equations for each user, resulting in a total of $n/2 + 2\delta$ linearly independent equations for each codeword, will be enough to ensure correct decoding, allowing for $R<3/4$.
In the following we will construct LDPC codes that approach this limit with linear decoding complexity.

\section{LDPC Code Design} \subsection{LDPC Code Ensembles} LDPC codes are defined by a bipartite graph
where the transmitted bits are represented by VNs which are subject to local parity checks,
represented by CNs. We study random codes that are drawn uniformly at random from a given
ensemble, defined by the degree distribution of VNs and CNs. Specifically, a random graph code from the ensemble
is created by first assigning degrees to VN and CNs proportional to some degree distributions. Then the emanating
stubs (half-edges) of VNs and CNs are connected through a uniform random permutation (multi-edges are not explicitly forbidden). Finally the VNs are also permuted uniformly at random. We would like to
emphasize that it is important for our construction that the ensemble definition includes a random
permutation of the VNs. For memoryless single-user channels this is usually not necessary since
the error probability is invariant under permutation of VNs, and some works do not mention it for this reason, e.g., \cite{Ric2008}.
However, in the multiple-access case
correlations between VN degrees of neighboring nodes may introduce unwanted correlations in the
joint graph.

Let $L_i$ denote the fraction of nodes with degree $i$, $\lambda_i$ the fraction of edges that connect to degree $i$ VNs, and $\rho_i$ the fraction of edges that connect to degree $i$ CNs. We also define the corresponding power series $L(x):=\sum L_i x^i, \lambda(x) := \sum \lambda_i x^{i-1}$, and $\rho(x) := \sum \rho_i x^{i-1}$, and we
denote the corresponding ensemble as LDPC($\lambda,\rho$). 
\subsection{Message Passing Decoding}
We study the bit-error probability under BP decoding on the joint graph.The values of VNs $(v_{1,i},v_{2,i})$ are initialized with their know values if $y_i \neq 0$ and are initialized with the erased symbol $\epsilon$ if $y_i = 0$.  BP decoding on the joint graph can be realized by running two conventional single-user BP decoders on $(y_1,...,y_n)$ and $(y_{1+\tau},...,y_{n+\tau})$ respectively and exchanging information between them on $(y_{1+\tau},...,y_n)$. The information exchange is particularly simple for the BAC since $c_{1,i}$ fully defines $c_{2,i-\tau}$ given $y_i$. We denote the function nodes that enforce the channel constraint \eqref{eq:channel} as \emph{MAC nodes}. An example of a joint graph is depicted in 
\figref{fig:joint_graph} where triangles depict MAC nodes, squares are CNs, and circles are VNs. 
\begin{figure}[h]
\begin{center}
\includegraphics[width=0.5\columnwidth]{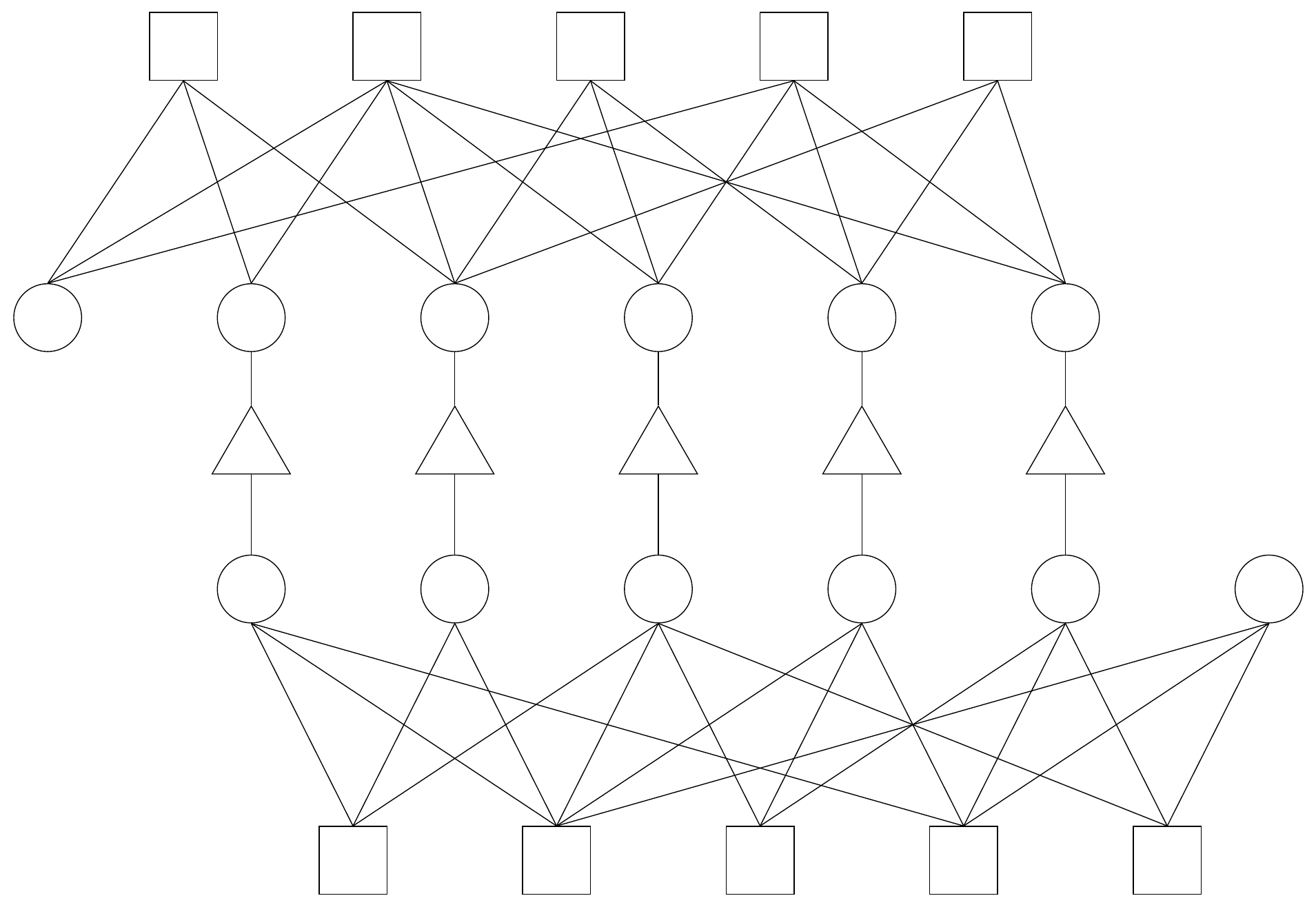}
\caption{Factor Graph for a UBAC with $\tau = 1$. Triangles denote MAC nodes, squares are CNs, circles are VNs. }
\label{fig:joint_graph}
\end{center}
\end{figure}
The single-user decoder can be run for multiple iterations before information exchange. Nonetheless, in this paper we only study the case where each iteration of the single-user decoders is followed by a message exchange through the MAC nodes. This decoder has $\mathcal{O}(n)$ complexity. 
\subsection{Coset Codes}
To simplify the analysis we consider the ensemble of cosets of LDPC codes where each code in this ensemble is specified by a graph $\mathcal{G}$ and a `dither' vector $\tilde{\mathbf{d}} \in \{0,1\}^n$ with its BPSK representation $\dv \in \{\pm 1\}^n$. The ensemble is then specified by a degree distributions pair $(\lambda(x), \rho(x))$ and the dither vector.  We consider the ensemble generated by randomly choosing VN and CN degrees according to the distribution pair $\lambda(x),\rho(x)$ followed by a random permutation between the left sockets and right sockets, and by choosing $\tilde{\dv}$ uniformly from $\{0,1\}^n$. 
Let $\mathcal{C}_{\mathcal{G},\tilde{\dv}}$ denote the coset code corresponding to a given $\mathcal{G}$ and $\tilde{\dv}$.
Let $\mathbf{G}$ and $\mathbf{H}$ denote the generator matrix and parity check matrix of the LDPC code, respectively, with a given $\mathcal{G}$ and $\tilde{\dv} = \mathbf{0}$. 
Then, $\mathbf{m} \in \mathcal{C}_{\mathcal{G},\tilde{\dv}}$ if and only if $\mathbf{H} \mathbf{m} = \mathbf{H} \tilde{\dv}$.

At the encoders, the bit sequences $\mathbf{b}_1$ and $\mathbf{b}_2$ are encoded into codewords $\mv_1$ and $\mv_2$, respectively, according to
\begin{eqnarray}
\mv_u = \mathbf{G} \mathbf{b}_u + \tilde{\dv}, \quad u\in\{1,2\}.
\end{eqnarray}
Note that both users share the same dither $\tilde{\dv}$. Since the BPSK mapping is one-to-one, we can also express the addition of the dither as multiplication of $\cv_1,\cv_2$ with $\dv$, resulting in the channel output
\[
y_i = c_{1,i} d_i + c_{2,i-\tau} d_{i-\tau}.
\]
Since $\dv$ is chosen as part of the code design, it is known at the receiver and its effect can be easily incorporated into the message passing rules. The analysis in Section \ref{sec:de} will show that a randomly chosen dither will be good for any code and all codeword combinations with probability approaching 1 as $n\to\infty$. 
\begin{remark}
    Note that the constructed LDPC codes are not strictly linear but affine. Nonetheless, they can be encoded with a linear encoder followed by a common offset. Besides, numerical results suggest that the error probabilities stay unchanged when no dithering is used. As such, the dither is mainly used as an analytic tool here.
\end{remark}
\section{Density Evolution Analysis}
\label{sec:de}
We next track the fraction of erased edges through the iterations averaged over the code and dither ensemble as $n\to\infty$. 
Let $x_l$ be the probability that a message from a variable node to a check node is erased, $y_l$ the probability that a message from a check node to a variable node is erased, $w_l$ the probability that a message from a variable node to a MAC node is erased, and $z_l$ the probability that a message from a MAC node to a variable node is erased.
The subscript $l$ refers to the $l$-th
iteration.
The passed messages are visualized in \figref{fig:BP_notation}.
\begin{figure}[h]
\begin{center}
\includegraphics[width=0.7\columnwidth]{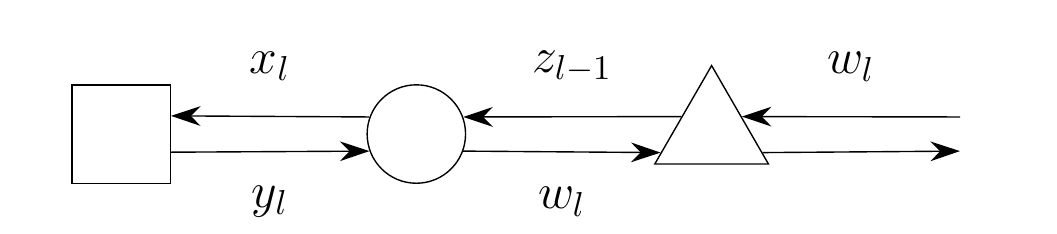}
\caption{Fraction of erased messages between VNs, CNs and MAC nodes.}
\label{fig:BP_notation}
\end{center}
\end{figure}

Assuming that the depth $l$ neighborhood of each node is a tree, we can derive a recursion for the evolution of the above parameters as follows.
Begin with initial conditions $y_0 = 1, x_0 = 1, z_0 = 1/2$
\begin{eqnarray}
x_{l+1} & = & z_l \lambda(y_l) \\
y_{l+1} & = & 1- \rho \left(1 - x_{l+1} \right) \\
w_{l+1} & = & L\left(y_{l+1}\right)\\
z_{l+1} & = & \frac{1}{2} w_{l+1}.
\end{eqnarray}
These equations are obtained by following the basic message passing rules. An edge from a degree $i$ VN
to a CN is erased if all incoming edges are erased. The VN has a total of $i-1$ incoming edges from other CNs which are independently erased with probability $y_l$ and one
incoming edge from a MAC node which is erased with probability $z_l$, resulting in an erasure probability $z_l y_l^{i-1}$. Averaging over all VN degrees gives the expression for
$x_{l+1}$. The other equations are derived similarly. The factor $1/2$ in $z_{l+1}$ arises since the value of each MAC node is independently erased with probability $1/2$. Note that this is only true because of the symmetrization by the dither.

By performing some standard substitutions, we end up with the following scalar recursion:
\begin{equation}
\begin{split}	
x_{l+1} %
& =  \frac12 L\left(1 - \rho \left(1 - x_l \right) \right) \lambda \left(1 - \rho \left(1 - x_l \right) \right).
\end{split}
\label{eq:xlrecursion}
\end{equation}

Likewise, we can obtain the following recursion on $y_l$:
\begin{equation}
\label{eq:ylrecursion}
    y_{l+1} = 1 - \rho \left(1 - \frac12 L(y_l) \lambda(y_l) \right).
\end{equation}

The probability that a bit remains erased at the end of iteration $l+1$ is given by
\begin{equation}
    p_{l+1} = z_{l} L(y_{l+1}),
    \label{eq:de_pe}
\end{equation}

where $(p_l)_{l=1,2,...}$ is a deterministic sequence of numbers. Our main theorem below shows that the BER of a randomly chosen code with a random dither sequence after $l$ decoding iterations concentrates tightly around $p_l$.
Let 
\begin{equation}
\begin{split}
	P_b(\dv,\cv, l) &:= P_b(\cv,\mathcal{G},n,l,\dv,\tau)\\
                        &= \frac{1}{2n}\sum_{i=1}^{2n} \EE[\mathds{1}\{v_i^l = \epsilon\}|\mathcal{G},\dv]
\end{split}
\end{equation}
be the BER (fraction of erased VNs) at blocklength $n$ after $l$ iterations for a given code $\mathcal{G} \in\text{LDPC}(\lambda,\rho)$ and codeword pair $\cv = (\cv_1,\cv_2)$. Also let $\bar{P}_b(\dv,l) = \frac{1}{|\mathcal{C}|^2}\sum_{\cv}P_b(\dv,\cv,l)$ denote the average BER.
Then the following holds:
\begin{theorem}
\label{thm:de}
	As $n\to\infty$, for any $\tau \in [1:\tau_\text{max}]$
	\begin{equation}
		\PP_{\mathcal{G},\dv}(|\bar{P}_b(\dv,l) - p_l|>\lambda) \to 0	
	\end{equation}
	for any $\lambda > 0$.
 \hfill $\square$
\end{theorem}
\begin{proof}
    The proof is given in Appendix~\ref{appendix:de}.
\end{proof}

\section{Optimization}
\label{sec:opt}
We can use the DE equations to optimize the degree distributions.  Specifically, define
\begin{equation}
	f_\rho(y) = y -1 + \sum_{i=2}^{r_\text{max}} \rho_i \left(1 - \frac{1}{2}L(z(y)\lambda(y)\right)^{i-1}	
	\label{eq:rho_linear}
\end{equation}
where $r_\text{max}$ is the maximal CN degree.
For fixed $\lambda$, \eqref{eq:rho_linear} is linear in $\rho_i$ and gives rise to the linear program:
\begin{equation}
\begin{split}
    \min_\rho \quad & \sum_{i} \frac{\rho_i}{i}\\
    \textrm{s.t.} \quad & \rho_i\geq 0; \sum_{i}\rho_i = 1;f_\rho(y)>\delta\ \forall y \in (0,1)
 \end{split}
\end{equation} 
where $\delta\geq 0$ is a slack variable.
For fixed $\rho$, \eqref{eq:xlrecursion} results in an optimization problem with linear objective and quadratic constraints.
Details on the quadratic program are given in Appendix \ref{appendix:opt}.
Unfortunately, it can be shown that the constraints are not positive semidefinite. Therefore, the problem is not convex in general and a solver is not guaranteed to converge to the optimal solution. Nonetheless, we find that general purpose quadratic solvers lead to good results and we are able to empirically find degree distributions that achieve rates close to the BAC capacity by alternating optimization of $\rho$ and $\lambda$.
To find distributions which can be decoded in a reasonable amount of iterations and are robust to finite length fluctuations we follow \cite[Sec. VII]{Sho2006} and set the slack variable to $\delta = c/\sqrt{n}$. The parameter $c$ is set empirically. Higher $c$ will result in lower rates but less required decoding iterations.
\subsection{Error-Floor Analysis}
\label{sec:cycles}
In single-user LDPC ensemble constructions, degree one VNs are usually avoided because they prevent the BER (and the BLER) from going to zero. Indeed, when two degree one VNs connect to the same CN, they create a low-weight stopping set that cannot be recovered, even by an ML decoder.
However, for the two-user frame-asynchronous case, under certain circumstances, the presence of degree one VNs does not prevent the BLER from going to zero as $n\to\infty$. As we shall see, this implies that we can increase the rates in the finite-blocklength regime without introducing error floors by introducing a small fraction of degree one VNs.

In the joint graph, degree one VNs can be recovered through the MAC nodes, even if they connect to the same CN. In the following theorem we provide a bound on the probability that a randomly chosen graph with a fraction $L_1$ of degree one VNs has a $4K$-sized stopping set, consisting of just degree one VNs. The case $K=1$ is depicted in \figref{fig:cycle}.
	
\begin{figure}[h]
\begin{center}
\includegraphics[width=0.4\columnwidth]{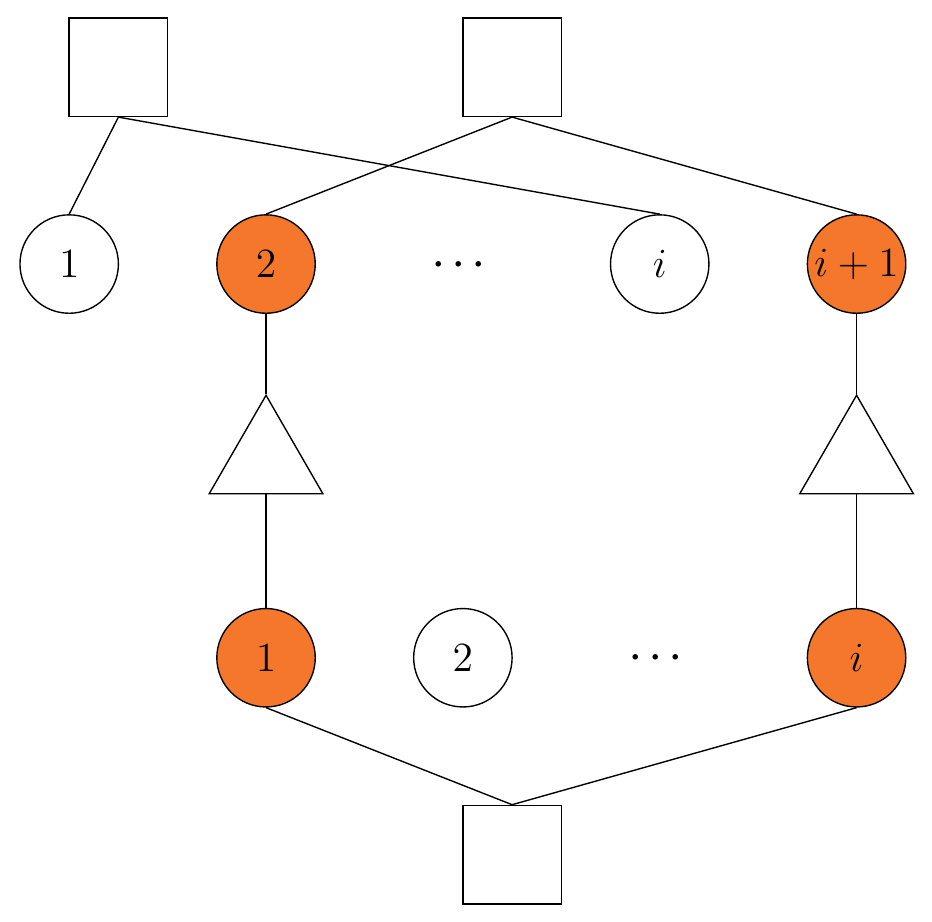}
\caption{Stopping set of size $4$ in a joint graph for $\tau=1$}
\label{fig:cycle}
\end{center}
\end{figure}
	
\begin{theorem}
\label{thm:4K-cycle}
	The probability that a random code from the ensemble LDPC($\lambda,\rho$) results in a joint graph that has no stopping sets of size $\leq 4K$ created by just degree one VNs for all $\tau \in [1:\tau_\text{max}]$
	can be bounded from below by
 \begin{equation}
      1 - \frac{\tau_\text{max}}{2}\sum_{k=1}^K\left(\frac{L_1^{2}}{1-R}\right)^{k}\frac{1}{2k} - \mathcal{O}\left(\frac{K}{n^{2}}\right)
      \label{eq:4K-cycles}
 \end{equation}
\hfill $\square$
\end{theorem}
\begin{IEEEproof}
See Appendix \ref{appendix:error_floor}.
\end{IEEEproof}
The above theorem also implies the following result on the BLER.
\begin{theorem}
\label{thm:BLER_expurg}
If $L_1$ is sufficiently small compared to $\tau_\text{max}$ such that \eqref{eq:4K-cycles} is strictly larger than zero, there exists a constant fraction of codes in the ensemble with a vanishing BLER. \hfill $\square$
\end{theorem}
\begin{IEEEproof}
See Appendix \ref{appendix:error_floor}. 
\end{IEEEproof}
Theorem \ref{thm:BLER_expurg} shows that error floors can be avoided by resampling the code until one is found where the joint graph contains no $4K$-sized stopping sets
for a desired range $\tau_\text{max}$. 
It is necessary that $\tau_\text{max}$ is small compared to
$L_1^2/(1-R)$.\footnote{This can be the case in applications where synchronization is possible but a small shift $\tau$ is deliberately introduced at the transmitters.} 
Note that even if $4K$-sized stopping sets of degree one VNs exist, they only result in bit-errors if all VNs in the set are
erased (i.e., users transmit different symbols), which happens with probability $2^{-4K}$. Therefore it may not be necessary to expurgate
these sets for large $K$, depending on the desired BLERs.  Besides, fixed length stopping sets result in a
number of bit-errors which does not scale with $n$. As such, they could also be corrected by
adding an outer code with rate approaching 1 as $n\to\infty$. See also the discussion in
\cite{Ric2001a}. 

\section{Numerical Results} \bgroup \def\arraystretch{1.2}
\begin{table}
\scriptsize
\begin{center}

\begin{tabular}{ |c|c|c|c|} 
\hline
     		& Code 1 & Code 2 & Code 3 \\
\hline
\hline 
$L_1$		& 0.376 & 0.560 & 0.444 \\
$L_2$ 		& 0.594 & 0.371 & 0.445 \\
$L_5$ 		& 0.014 &  		& 		\\
$L_6$ 		& 0.016 &  		& 		\\ 
$L_7$ 		&  		& 0.061 & 		\\ 
$L_8$ 		&  		& 0.008 & 0.111 \\ 
$R_4$ 		& 0.586 & 0.128 & 0.323 \\ 
$R_5$ 		& 0.188 & 0.582 & 0.489 \\
$R_{10}$ 	& 0.227 & 0.290 & 		\\ 
$R_{20}$ 	& 		&  		& 0.188 \\
\hline
\hline 
Design Rate 				& 0.689 & 0.716 & 0.733\\
\hline
Mean Iterations  	& 30 & 30 & 100 \\  
\hline
\end{tabular}
\caption{Degree distributions for three codes at different rates.}
\label{tab:dd}
\end{center}
\end{table}
\egroup
Table \ref{tab:dd} shows some degree distributions obtained using the optimization procedure given in Section \ref{sec:opt}. The slack variable $\delta$ was adjusted empirically to find codes that work with small blocklength and a reasonable number of required iterations.
The erasure probability for Code 2 in Table \ref{tab:dd} predicted from DE is shown in \figref{fig:de} together with some random decoding realizations with blocklength $n=5\cdot 10^4$.
The empirical block error rate (BLER) of the codes in Table \ref{tab:dd} is shown in \figref{fig:bler_fixed} for a fixed delay $\tau = 1$. For the code construction we choose a random sample from the permutation ensemble and we check if it contains $4K$-stopping sets up to $K=3$. If it does, we sample again. The number of required samples is typically less than 10 for Code 2 and between zero and two for Codes 1 and 3. We can see in \figref{fig:bler_fixed} that the resulting codes do not show an error floor. 
The case with random delay $\tau\in [1:\tau_{\max}]$ is explored in \figref{fig:ber_random}. We choose $\tau_\text{max} = 100$ for Code 1 and $\tau_\text{max} = 500$ for Codes 2 and 3. The reason for choosing a smaller $\tau_\text{max}$ for Code 1 is that for $n<1000$, a delay of several hundred symbols is a significant fraction of the blocklenght, in which case the number of symbols where both codewords collide is rather small and hence, the BER is small, too. This effect also explains the non-monotonic behavior of the BER for Code 2. Note that both BLER and BER are limited by $1/\tau_\text{max}$ because $\tau = 0$ will always result in a block error. As expected from the analysis in Section \ref{sec:cycles}, the codes exhibit an error floor due to short length stopping sets caused by degree one VNs and therefore the corresponding BLERs do not vanish. We can observe in the simulations that for large enough $n$, block errors are caused almost exclusively by $4$ remaining bit-errors for Code 1 and 3, while Code 2 also occasionally exhibits $8$ or $12$ remaining bit-errors. Thus, a high-rate outer code would be sufficient to resolve the remaining bit-errors in this case. For example, a BCH code would suffice with minimum distance $8$ or $24$, respectively.    

\begin{figure}[]
	\begin{center}
	\includegraphics[width=.95\columnwidth]{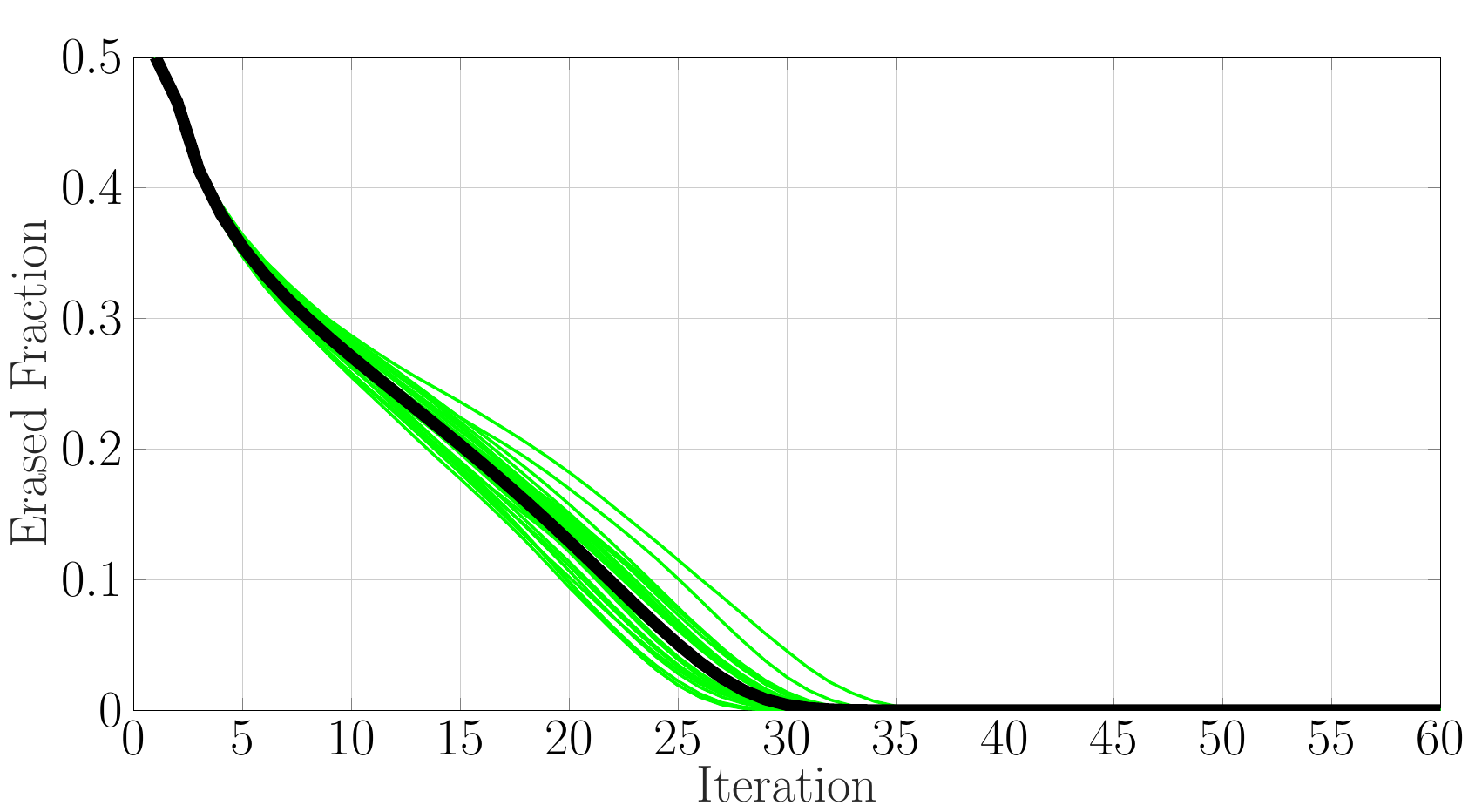}
	\caption{Erased fraction of VNs as a function of the number of iterations for Code 2. The black thick line represents the erasure probability from DE. The thin lines are sample paths for $n=5\cdot 10^4$.}
	\label{fig:de}
	\end{center}
\end{figure}
\begin{figure}[]
	\begin{center}
	\includegraphics[width=\columnwidth]{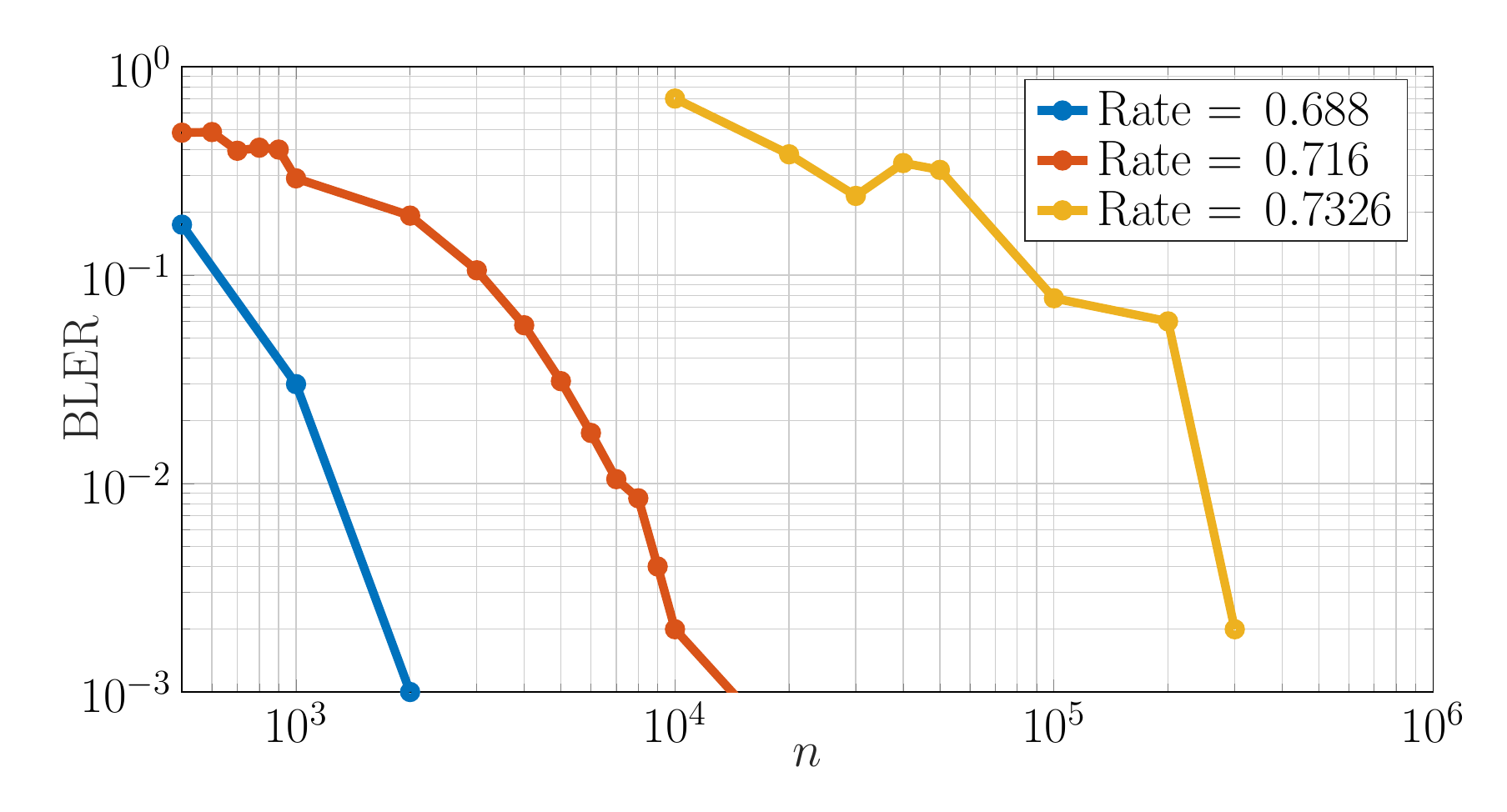}
	\caption{BLER as a function of $n$ for $\tau = 1$.}
	\label{fig:bler_fixed}
	\end{center}
\end{figure}
\begin{figure}[]
	\begin{center}
	\includegraphics[width=.95\columnwidth]{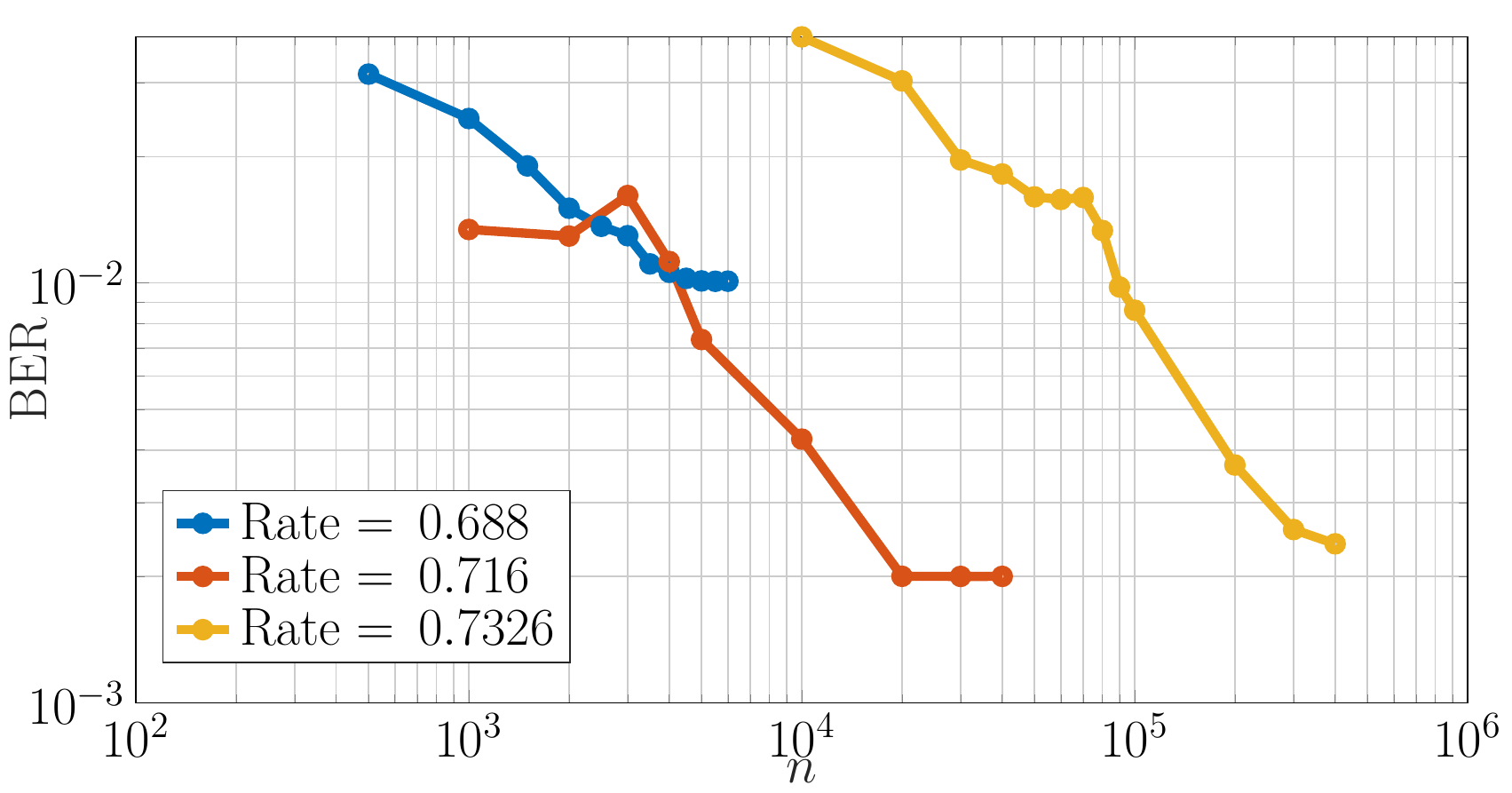}
	\caption{BER as a function of $n$ for random $\tau \in [0,\tau_\text{max}]$, with $\tau_\text{max} = 100$ for Code 1, and $\tau_\text{max} = 500$ for Code 2 and 3.}
	\label{fig:ber_random}
	\end{center}
\end{figure}
\printbibliography
\clearpage
\appendices
\section{Proof of Theorem \ref{thm:rlc}}
\label{appendix:rlc}
\begin{proof}
    We start by reformulating the decoding problem on the frame-asynchronous UBAC in terms of the parity check matrix $\Hm \in \mathbb{F}_2^{n-k\times n}$ of a linear code. 
    Let $\mv_1,\mv_2$ be two codewords, i.e., $\Hm \mv_1 = \Hm \mv_2 = 0$, and let both codewords be transmitted through the BAC \eqref{eq:channel} for some fixed $\tau$.
    Let $\mathcal{E} = \{i: y_i = 0\}$ and denote the shifted set by $\mathcal{E}-\tau := \{i: y_{i+\tau} = 0\}$. For $i\in\mathcal{E}$ we have $m_{1,i} = 1 - m_{2,i-\tau}$. Let $\Hm_\mathcal{E}$ denote the sub-matrix of $\Hm$ with column indices in $\mathcal{E}$ and, analogously, $\mv_{u,\mathcal{E}}$ be the restriction of $\mv_u$ to the set $\mathcal{E}$.
    Note that on $\overline{\mathcal{E}} = [n]\setminus \mathcal{E}$ the entries of $\mv_1$ are known and similarly on $\overline{\mathcal{E}-\tau}$ the entries of $\mv_2$ are known. Therefore, we can compute the two syndromes:
    \begin{equation}
    \begin{split}
        \sv_1 &= \Hm_{\overline{\mathcal{E}}}\mv_{1,\overline{\mathcal{E}}} \\
        \sv_2 &= \Hm_{\overline{\mathcal{E}-\tau}}\mv_{2,\overline{\mathcal{E}-\tau}}
    \end{split}
    \end{equation}
    and $\mv_1$ satisfies the two constraints:
    \begin{equation}
    \begin{split}
        \Hm_\mathcal{E} \mv_{1,\mathcal{E}} &= \sv_1 \\
        \Hm_{\mathcal{E}-\tau} \mv_{1,\mathcal{E}} &= \Hm_{\mathcal{E}-\tau} (\mathbf{1}-\mv_{2,\mathcal{E}-\tau})\\ &= \Hm_{\mathcal{E}-\tau}\onev - \sv_2 =: \tilde{\sv}_2
        \label{eq:pc}
    \end{split}
    \end{equation}
    We can define $\tilde{\Hm} = [\Hm\ \zerov_\tau;\zerov_\tau\ \Hm] \in \mathbb{F}_2^{2(n-k)\times (n+\tau)}$. 
    With this \eqref{eq:pc} can be written as $\tilde{\Hm}_\mathcal{E} \mv_{1,\mathcal{E}} = [\sv_1;\tilde{\sv}_2]$.
    This equations can be solved if the rank of $\text{rank}(\tilde{\Hm}_\mathcal{E}) > |\mathcal{E}|$.
    
    Now let the entries of the parity check matrix $\Hm$ be Bernoulli(1/2) i.i.d. and define $r=n-k$. 
    For some arbitrary erasure set $\mathcal{E} \subset [\tau:n]$ of size $d$ we compute the probability $P_{d}$ that a sub-matrix $\tilde{\Hm}_\mathcal{E}$ of $\tilde{\Hm}$ of size $2r\times d$ has rank $d$. Note that this probability is well defined since it does only depend on the size of $\mathcal{E}$ but not on the actual set.  
    The complications in the proof, compared to standard techniques, arise from the fact that $\tilde{\Hm}$ may contain the same vectors in top and bottom half, in which case we cannot assume anymore that they are independent. 
    We can bound $P_d$ as follows:
    \begin{equation}
        P_{d} \geq \prod_{k=1}^d\left(1 - \frac{2^{k+1}}{2^{2r}}\right)
        \label{eq:p_rank}
    \end{equation}
    To get the bound we compute the smaller probability $\tilde{P}_d$ that $\Hm_\mathcal{E}$ has rank $d$ \emph{and} the following condition is fullfilled:
    \begin{enumerate}[label=\roman*), ref=\roman*]
        \item\label{item:cond} Non of the top half vector from $\Hm_\mathcal{E}$ are in the column span of the bottom half $\Hm_{\mathcal{E}-\tau}$.
    \end{enumerate}
    We compute $\tilde{P}_d$ recursively by adding the indices in $\mathcal{E}$ in increasing order: Let $\mathcal{E}_k$ denote the sub-set of $\mathcal{E}$
    with only the first $k$ indices. 
    Assume the columns of $\tilde{\Hm}_{\mathcal{E}_{k-1}}$ are linearly independent and condition \ref{item:cond}) is satisfied. If one column $\tilde{\hv}_i := [\hv_i;\hv_{i-\tau}]$ is added, the resulting set will be linearly dependent if 
    $\{\tilde{\hv}_i \in \text{span}(\tilde{\Hm}_{\mathcal{E}_{k-1}-\tau})\} := \mathcal{I}_1$. In addition, condition \ref{item:cond}) will be broken if
    $\{\hv_i \in \text{span}(\Hm_{\mathcal{E}_{k-1}-\tau})\} := \mathcal{I}_2 $ happens.
    $\mathcal{I}_1$ can be further decomposed into the two disjoint events $\mathcal{I}_{1,1} = \mathcal{I}_1 \cap \{i-\tau \in \mathcal{E}_{k-1}\}$
    and $\mathcal{I}_{1,2} = \mathcal{I}_1 \cap \{i-\tau \notin \mathcal{E}_{k-1}\}$. The conditional probability of $\mathcal{I}_{1,1}$ is zero due to the assumption that condition \ref{item:cond}) is fullfilled. On the other hand, if $i-\tau \notin \mathcal{E}_{k-1}$ then $\hv_{i-\tau}$ is independent of
    $\tilde{\Hm}_{\mathcal{E}_{k-1}}$ and 
    the probability that $\mathcal{I}_{1,2}\cap \mathcal{I}_2$ happens is $(1 - (2^{k-1} + 2^{k})/2^{2r})$ since there are $2^{k-1}$ binary vectors in the span of $\Hm_{\mathcal{E}_{k-1}}$ and $2^k$ vectors in the span of $[\Hm_{\mathcal{E}_{k-1}-\tau},\hv_{i-\tau}]$. So we can bound $\tilde{P}_d$ as
    \begin{equation}
    \begin{split}
    \tilde{P}_d &\geq \left(1 - \frac{2^{k+1}}{2^{2r}}\right)\tilde{P}_{d-1}
    \end{split}
    \end{equation}
    which proves \eqref{eq:p_rank}.
    
    For a fixed parity check matrix $\Hm$, the probability of decoding the first codeword wrong, averaged over all codeword pairs, is given by 
    \begin{equation}
    \begin{split}
        P_{e,\Hm} &= 1 - \sum_{d=1}^{n-\tau} \PP(|\mathcal{E}| = d)\mathds{1}(\text{rank}(\tilde{\Hm}_\mathcal{E}) = d)
        \label{eq:pe_fixedH1}
    \end{split}
    \end{equation}
    Let $n':=n-\tau$ and $d_\text{max} = n'/2 + \delta n'$.
    We can write $|\mathcal{E}| = \sum_{i=1}^{n'}X_i$ where we define $X_i := \mathds{1}(c_{1,i}=-c_{2,i-\tau})$.
    It holds that $\EE[X_i] = 1/2$, $\text{Var}[X_i] = 1/4$,
    \footnote{This holds for all typical codes, that is those which have a marginal bit distributions close to Bernoulli(1/2). It can be shown easily that all but an exponentially small fraction of random parity check codes are typical. So we can restrict parity check matrices to be typical.}
    and the $X_i$ are \emph{pairwise} independent
    \footnote{Again, this is true for all but an exponentially small fraction of parity check matrices. This can be seen by bringing $\Hm$ into systematic form. Then the first $k$ data bits are clearly independent and two of the $n-k$ parity check bits are dependent if and only if they are sums of the exact same set of bits.}. Therefore $\text{Var}(|\mathcal{E}|) = n/4$ and Chebychev's inequality shows that for any $\delta >0$
    \begin{equation}
        P\left(\left||\mathcal{E}| - \frac{n}{2}\right|>\delta n\right) \leq \frac{1}{4\delta n}.
    \end{equation}
    Since $\mathds{1}(\text{rank}(\tilde{\Hm}_\mathcal{E}) \geq d)$ is non-increasing in $d$ we have
    \begin{equation}
    \begin{split}
        P_{e,\Hm} 
        &\leq 1 - \mathds{1}(\text{rank}(\tilde{\Hm}_\mathcal{E}) = d_\text{max}) P(|\mathcal{E}| \leq d_\text{max} )\\
        &\leq 1 - \mathds{1}(\text{rank}(\tilde{\Hm}_\mathcal{E}) = d_\text{max}) + o_n(1).
        \label{eq:pe_fixedH}
    \end{split}
    \end{equation}
    Note that the channel is noiseless. Therefore, if one codeword is correctly recovered the second one can be obtained by subtracting the first.
    Also, it is irrelevant which codeword we attempt to decode since both users share the same $\tilde{H}_\mathcal{E}$.
    For simplicity we assume that $d_\text{max} = (n-\tau + \delta)/2$ is an integer. Averaging \eqref{eq:pe_fixedH} over the code ensemble we get
    \begin{equation}
    \begin{split}
        P_e 
        &\leq 1-\tilde{P}_{(n-\tau+\delta)/2} + o_n(1)\\
        &= \frac{n-\tau+\delta}{2}2^{n(1/2 - 2(1-R))} + o_n(1)
    \end{split}
    \end{equation}
\end{proof}
An alternative proof for the slightly different linear code ensemble of iid random generator matrices $\Gm \in \mathbb{F}_2^{k\times n}$ can be sketched as follows:
Let $(x_i)' = x_{i-\tau}$ for $i>\tau$ denote the left-shift of a vector entry by some fixed $\tau$ and let $\uv_1,\uv_2 \in \mathbb{F}_2^{k}$ be the transmitted bit sequences. Then the channel output reads as
\begin{equation}
    \yv = (\uv_1\Gm) + (\uv_2\Gm)'
\end{equation}
W.l.o.g. we can choose $\uv_1 = \ev_1$ and $\uv_2 = \ev_2$. For arbitrary $\uv_1,\uv_2$ we can find a basis which has $\uv_1,\uv_2$
as first two basis vectors and work in the new basis. Since the distribution of $\Gm$ is invariant under basis change this does not affect the error probability.

Conditioned on the two rows $\gv_1,\gv_2$ the error probability is given by
\begin{equation}
    P_{e,\Gm} = \PP(\bigcup_{\vv_1,\vv_2} \{(\vv_1\Gm) + (\vv_2\Gm)' = \gv_1 + \gv_2'\}|\gv_1,\gv_2)
\end{equation}
We partition the space of possible sequences $\vv_1,\vv_2$ into two sets $\mathcal{A}$ and $\mathcal{A}^c$ such that sequences in $\mathcal{A}$ are zero in the first two positions. Within $\mathcal{A}$, $\vv_1\Gm,\vv_2\Gm$ are independent of $\gv_1,\gv_2$. Furthermore, we distinguish two cases. First, that $\vv_1$ and $\vv_2$ both contain at least one unique bit. In this case we can treat them as independent vectors and bound the error probability, averaged over $\Gm\setminus [\gv_1,\gv_2]$ as
\begin{equation}
    \begin{split}
    P_{e} &\leq 2^{2(k-2)}\PP(b_1 + b_2 = 0)^{|\mathcal{Y}_0|}\PP(b_1 + b_2 = 1)^{|\mathcal{Y}_1|} \\
    &\ \PP(b_1 + b_2 = 2)^{|\mathcal{Y}_2|}
    \end{split}
\end{equation}
where $b_1,b_2$ are independent Bernoulli(1/2) distributed bits and $\mathcal{Y}_l = \{i:y_i = l\}$. Averaging over $\gv_1,\gv_2$ gives
\begin{equation}
\begin{split}    
    P_{e} &\leq 2^{2(k-2)}\left(\frac{1}{4}\right)^{n/4}\left(\frac{1}{2}\right)^{n/2}\left(\frac{1}{4}\right)^{n/4}  + o_n(1)\\
    &= 2^{n(2R - 1.5)} + o_n(1)
    \label{eq:pe_gen_indep}
\end{split}
\end{equation}
In the second case, where $\vv_2$ is of the form $\vv_2 = \vv_1 \oplus \uv$ for some vector $\uv$ that is independent of $\vv_1$.
We get probabilities of the form $\PP(b_1 + (b_1\oplus b_2) = j)$ for $j\in \{0,1,2\}$, leading to (we skip the intermediate steps):
\begin{equation}
    P_{e} \leq 2^{2(k-2)}\left(\frac{1}{4}\right)^{n}  + o_n(1)\\
    = 2^{n(2R - 2)} + o_n(1)
\end{equation}
It remains to estimate the error probabilities in $\mathcal{A}^c$.
First note that whenever $\vv_1,\vv_2$ both contain at least one unique bit they can be treated as independent and we get the same bound as \eqref{eq:pe_gen_indep}. Also, if only one of them contains a random vector, e.g., $\vv_1 = \gv_1, \vv_2 = \gv_2 \oplus \uv$, then there is at most one values of $u_i$ for each $i$ which replicates the channel values. Resulting in 
\begin{equation}
    P_{e} \leq 2^{k-2}\left(\frac{1}{2}\right)^{n}  + o_n(1)\\
    = 2^{n(R - 1)} + o_n(1)
    \end{equation}
This leaves only $15$ possible cases, most of which are trivial or can be reduced by symmetry to one of the following 4 non-trivial cases:
\begin{itemize}
    \item Case 1: $\vv_1 = \gv_1 \oplus \uv, \vv_2 = \gv_2 \oplus \uv$\\
    We explicitly write down the equations that need to be satisfied for an error to occur (wlog for $\tau=1$):
    \begin{equation}
        g_{1,i} \oplus u_{i} + g_{2,i-1} \oplus u_{i-1} = g_{1,i} + g_{2,i-1}
    \end{equation}
    This can only be satisfied if $u_{i} = u_{i-1}$ which happens with probability $1/2$. Therefore we can bound the error probability
    \begin{equation}
    P_{e} \leq 2^{k-2}\left(\frac{1}{2}\right)^{n}  + o_n(1)\\
    = 2^{n(R - 1)} + o_n(1)
    \end{equation}
    \item Case 2: $\vv_1 = \gv_1 \oplus \gv_2 \oplus \uv, \vv_2 = \uv$\\
    In this case there are some channel values that cannot be replicated. E.g. $g_{1,i} = g_{2,i} = 0$ and $g_{2,i-1} = 1$, then $y_i = 1$, but $g_{1,i}\oplus g_{2,i} = 0$. So neither values if $u_i$ can replicate the channel output. Therefore $P_e = 0 + o_n(1)$. 
    \item Case 3: $\vv_1 = \gv_1 \oplus \gv_2 \oplus \uv, \vv_2 = \gv_1 \oplus \uv$\\
         Similar to Case 2, giving $P_e = 0 + o_n(1)$ similar to Case 2.
    \item Case 4: $\vv_1 = \gv_1\oplus \uv, \vv_2 = \uv$
    If $y_i = 1$ both $(u_i,u_{i-1}) = (0,1)$ and $(u_i,u_{i-1}) = (1,0)$ result in the correct channel output for both $g_{1,i} =0$ and $g_{1,i} =1$
    \begin{equation}
    P_{e} \leq 2^{k-2}\left(\frac{1}{4}\right)^{n/2}\left(\frac{1}{2}\right)^{n/2}  + o_n(1)\\
    = 2^{n(R - 1.5)} + o_n(1)
    \end{equation}
\end{itemize}
The most restricting constraint is \eqref{eq:pe_gen_indep}, allowing for any $R < 3/4$.
\section{Proof of Theorem \ref{thm:de}}
\label{appendix:de}
The outline of the proof is as follows:

Lemma~\ref{cor:cw_conc} shows that $P_b$ with fixed dither concentrates around the dither average.
Corollary~\ref{thm:dither_conc} shows that $P_b$ for a fixed codeword pair concentrates around the average over all codeword pairs. 
Lemma~\ref{thm:indep} establishes that the dither average is independent of the transmitted codewords. Lemma~ \ref{thm:treelike} states that the computation tree for each VN is with high probability tree-like for a fixed depth $l$ as $n\to\infty$.
Finally, we argue that that $P_b$ for any fixed random graph concentrates around the ensemble average, which concludes the proof of Theorem~\ref{thm:de}.

The proof will make repeated use of Azuma-Hoeffding's inequality \cite{Hoe1963,Azu1967}
applied to so called Doob martingales, which are conditional expectations of the form 
\begin{equation}
    Y_i = \EE[f(X_1,...,X_n)|X_1=x_1,..,,X_i=x_i]
\end{equation}
for some function $f$ and a (not necessarily iid) sequence of RVs $(X_i)_{i=1,...,n}$. It holds that $Y_0 = \EE[f]$, $Y_n = f(x_1,...,x_n)$, and
\begin{theorem}[Azuma-Hoeffding for Doob Martingales]
	Suppose that $|Y_k - Y_{k-1}|\leq d_k$ for a sequence $(d_k)_{k=1,...,n}$ of non-negative reals. Then for $\lambda>0$ it holds
	\begin{equation}
		\label{eq:azuma}
		\PP(|Y_n - Y_0|>\lambda) \leq 2\exp\left(-\frac{\lambda^2}{2\sum_{k=1}^n d_k^2}\right)
	\end{equation}
 \hfill $\square$
\end{theorem}

The next lemma shows that for two fixed transmitted codewords any randomly chosen dither sequence, with high probability, will result in a bit-error rate that is close to the bit-error rate averaged over all dither sequences. 
\begin{lemma}
    \label{thm:dither_conc}
	\begin{equation}
		\PP (|P_b(\dv,\cv_1,\cv_2) - \EE[P_b(\dv,\cv_1,\cv_2)]|>\lambda) \leq \exp(-C\lambda n)
		\label{eq:conc_d}
	\end{equation}	
	for some constant $C>0$ and any $\lambda>0$.
\end{lemma}
\begin{proof}
Define the Doob martingale $Y_i = \EE[P_b(\dv)|d_1,...,d_i]$. Since any dither value $d_i$ affects at most two VNs and all VNs included in their depth $l$ computation graphs, the number of affected VNs is upper bounded by a constant that does not scale with $n$. This constant can be bounded by the maximal VN and CN degrees in the graph as we will show later as part of the proof of Lemma~\ref{thm:treelike}. Therefore $Y_i$ has bounded increments and the concentration inequality \eqref{eq:conc_d} follows from \eqref{eq:azuma}. 	
\end{proof}
In fact, also the stronger statement holds, that a randomly chosen dither can be used for all codeword pairs $(\cv_1,\cv_2)$.
\begin{corollary}
\label{cor:cw_conc}
\begin{equation}
\begin{split}
  \PP \left(\left|\frac{1}{|\mathcal{C}|^2}\sum_{\cv_1,\cv_2}P_b(\dv,\cv_1,\cv_2) - \EE[P_b(\dv)]\right|>\lambda\right) \\
  \leq \exp(-C'\lambda n)  
\end{split}
\end{equation}	
for some constant $C'>0$ and any $\lambda>0$.
\end{corollary}
\begin{proof}
	First, it holds that
	\begin{equation}
		\PP_{\cv'_1,\cv'_2} (|\bar{P}_b(\dv)  - P_b(\dv,\cv'_1,\cv'_2)|>\lambda) \leq \exp(-C''\lambda n)
		\label{eq:conc_c}
	\end{equation}
	for some constant $C''>0$, any $\lambda>0$ and any fixed $\dv$. This can be shown  by applying Azuma-Hoeffding's inequality to the martingale that reveals $\cv'_1$ and 
	$\cv'_2$ component by component. Since each component affects at most a finite number of VNs in the depth $l$ neighborhood, the martingale has bounded increments.
	Furthermore,
	\begin{equation}
		\begin{split}
				&\PP_{\dv} (|\bar{P}_b(\dv)  - \EE[P_b(\dv)|>\lambda) \\
				&= \PP_{\cv'_1,\cv'_2,\dv} (|\bar{P}_b(\dv) - P_b(\dv,\cv'_1,\cv'_2)\\
    &+  P_b(\dv,\cv'_1,\cv'_2) - \EE[P_b(\dv)|>\lambda)\\
				&\leq \PP_{\cv'_1,\cv'_2,\dv} \left(|\bar{P}_b(\dv) - P_b(\dv,\cv'_1,\cv'_2)|>\frac{\lambda}{2}\right)\\
    &+ \PP_{\cv'_1,\cv'_2,\dv} \left(|P_b(\dv,\cv'_1,\cv'_2) - \EE[P_b(\dv)|>\frac{\lambda}{2}\right) \\
				&\leq 2\exp(-C'\lambda n)
		\end{split}
	\end{equation}
	where the last inequality follows by applying \eqref{eq:conc_d} and \eqref{eq:conc_c}, integrating, and setting $C' = \max\{C,C''\}/2$.
\end{proof}
The next lemma will show that the channel output, when averaged over the distribution of the dither, is iid and does not depend on the transmitted codewords $\cv_1,\cv_2$. Therefore, when evaluating $\EE[P_b(\dv)]$, we can assume that both users transmit the all-ones codeword.
\begin{lemma}
	\label{thm:indep}
	For any two transmitted codewords $\cv_1,\cv_2$ and any set $\mathcal{S} \subset [\tau+1:n]$, each symbol in the channel output is erased independently with probability $1/2$, i.e., 
	\begin{equation}
		\PP_\dv (\yv_\mathcal{S} = \zerov) = \left(\frac{1}{2}\right)^{|\mathcal{S}|}.
	\end{equation}
 \hfill $\square$
\end{lemma}
\begin{proof}
	Since $d_i,d_j$ are independent if $i\neq j$, we have $p(y_i = 0) = p(d_ic_{1,i} + d_{i-\tau}c_{2,i-\tau} = 0) = \frac{1}{2}$
	since $d_ic_{1,i}$ and $d_{i-\tau}c_{2,i-\tau}$ are independent and uniform over $\{-1,1\}$. Dependencies may occur only if $d_i$ is shared in multiple channel outputs. Note that only $y_i$ and $y_{i+\tau}$ include $d_i$. We can compute
	\begin{equation}
	\begin{split}
		&p(y_i = 0,y_{i+\tau} = 0) \\
            &= p(y_{i+\tau}= 0|y_i = 0)p(y_i = 0)\\
		&= \frac{1}{2}p(d_{i+\tau}c_{1,i+\tau} + d_{i}c_{2,i} = 0 |d_{i}c_{1,i} + d_{i-\tau}c_{2,i-\tau} = 0)\\
		&= \frac{1}{2}p(d_{i+\tau}c_{1,i+\tau} = d_{i-\tau}c_{2,i-\tau})\\
		&= \frac{1}{4}
	\end{split}
	\label{eq:y_i_indep}
	\end{equation}
	where the last inequality follows because $d_{i+\tau}$ and $d_{i-\tau}$ are independent. An arbitrary set $\mathcal{S}$ can be handled by using \eqref{eq:y_i_indep} repeatedly. 
\end{proof}

\begin{figure}[h]
\begin{center}
\includegraphics[width=0.65\columnwidth]{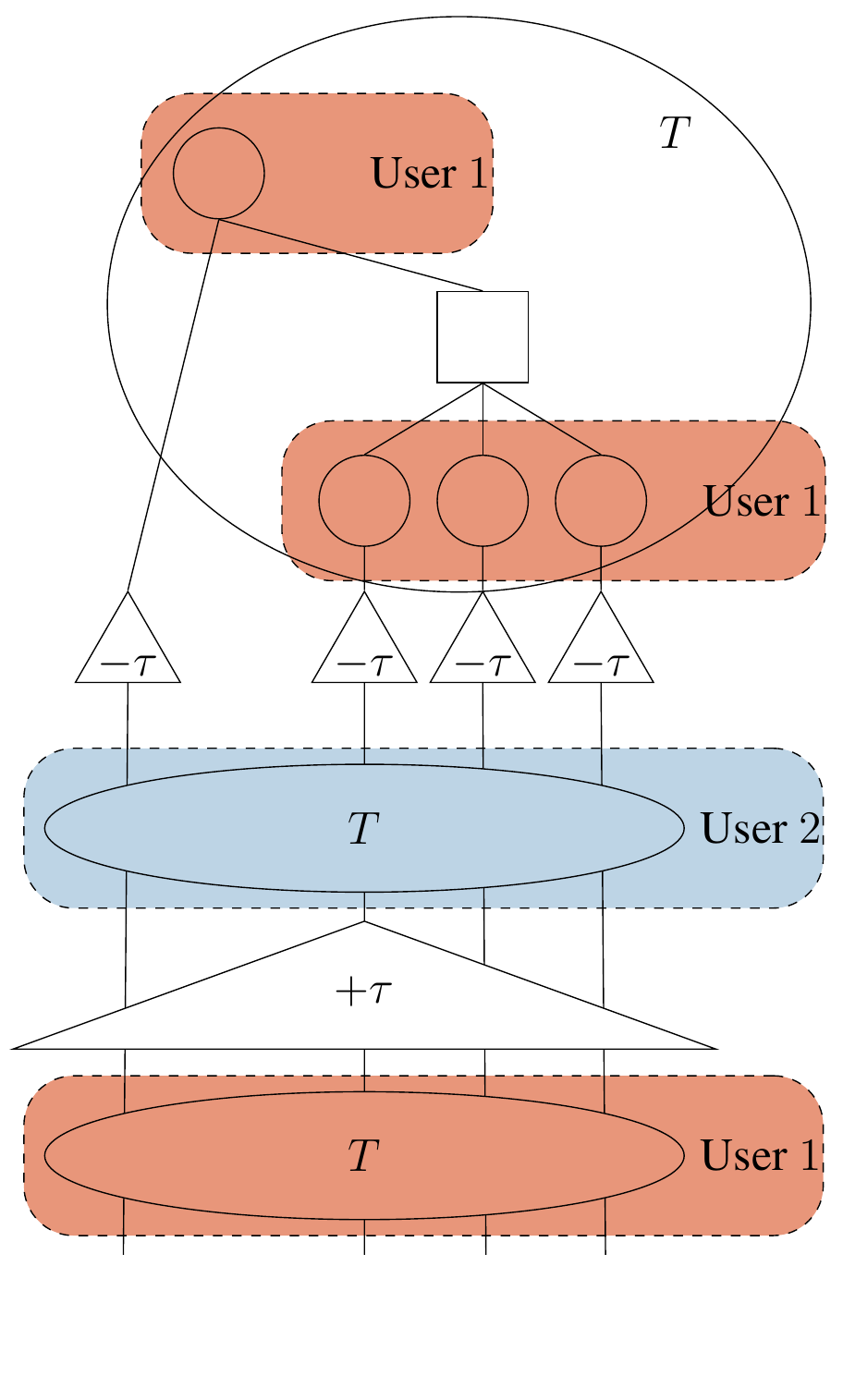}
\caption{Computation Graph, T denotes the basic LDPC computation tree with one VN connected to its adjacent CNs which in turn connected to their adjacent VNs.}
\label{fig:comp_tree}
\end{center}
\end{figure}

Next, we show that for a random code from LDPC($\lambda,\rho$) the depth $l$ computation graph rooted at a VN is a tree with high probability.
\figref{fig:comp_tree} depicts the structure of the computation graph of an arbitrary VN of user 1. The root node is connected to one MAC node (triangle) and a variable number of check nodes, which in turns connect to other variable nodes, which connect to other check nodes. After this point, at each additional iteration the structure of parity checks followed by MAC nodes is recursively repeated $l$ times. The number of leaves of each element, denoted by $T$, is an iid random variable whose distribution can be calculated from the left and right degree distributions. Here, we only need an upper bound on the number of leaves, which is given by $n_\text{max} = l_\text{max}r_\text{max} + 1$ where $l_\text{max}$ and $r_\text{max}$ are the maximal VN and CN degrees.
We next show that for a randomly chosen code from the ensemble LDPC($\lambda,\rho$) the probability that the nodes in a computation graph with root VN $i$, $i=1,...,2n$, contains only distinct VNs, and is therefore a tree, can be bound as follows.

Let us represent the VNs as two vectors $\vv_1,\vv_2 \in \{0,1\}^{n+\tau}$ with zero padding, i.e. $v_{1,j} = 0$ for $j\in [n+1:n+\tau]$
and $v_{2,j} = 0$ for $j\in [0,\tau]$. Let $\mathcal{V}_u$ denote the set of VNs for user $u$, $u\in \{1,2\}$. Due to the same-codebook constraint, the neighborhood of $v_{1,i}$ is the same as the neighborhood of $v_{2,i+\tau}$. 
The neighborhood of some fixed VN, without loss of generality in $\mathcal{V}_1$, at depth $t$ can be recursively expressed as follows. Let $\mathcal{N}^{2t}(\tau)$ denote the neighborhood of root VN $i$ (we drop the index $i$ for readability) at depth $t$ in the joint graph with offset $\tau$. We also drop the dependence on $\tau$ when immaterial.
 We split the neighborhood as 
$\mathcal{N}^{2t} = N^{2t}_1 \cup N^{2t}_2$ where $N^{2t}_1, N^{2t}_2$ denote the neighbors in $\mathcal{V}_1$ and $\mathcal{V}_2$ respectively. Let $N^{0}_1 = {i}$
be the root. We can describe the evolution of $N^{2t}_{u}$, $u=\{1,2\}$, with increasing depth as follows.
Define the shifted sets $\mathcal{N} \pm \tau := \{i: i = j\pm\tau, j\in\mathcal{N}\}$.
\begin{enumerate}
	\item $N^{2t+1}_1 = N^{2t}_1 \cup V^{t,1}$ where $V^{t,1}$ is the set of nodes in $\mathcal{V}_1$ that connect to $N^{2t}_1$ through a CN. 
	\item $N^{2t+1}_2 = N_1^{2t+1} - \tau$. The right hand side (rhs) is the set of nodes in $\mathcal{V}_2$ that connect to $N^{t+1}$ through MAC nodes.
	\item $N^{2t+2}_2 = N_2^{2t+1} \cup V^{t,2}$ where $V^{t,2}$ is the set of nodes in $\mathcal{V}_2$ that connect to $N^{2t+1}_2$ through a CN. 
	\item $N^{2t+2}_1 = N^{2t+2}_2 + \tau$. The rhs is the set of nodes in $\mathcal{V}_1$ that connect to $N^{2t+2}_2$ through MAC nodes.
\end{enumerate}

 Note that for a random code from the ensemble LDPC($\lambda, \rho$) the sets $V^{t,u}$ are random.

\begin{lemma}
	\label{thm:treelike}
	$$\PP(\mathcal{N}^{2T}(\tau) \text{ is not a tree for some } \tau \in [1:\tau_\text{max}]) \leq \frac{\gamma}{n}$$
	where $\gamma$ depends on $T,\lambda,\rho$ and $\tau_\text{max}$ but not on $n$. 
\end{lemma}
\begin{proof}
	The proof follows the structure of \cite{Lub1998,Ric2001a}. 
	Assume that the computation graph at iteration $t$, $t<T$, is a tree.\footnote{The first iteration is special as it has one more connection than the others, as depicted in \figref{fig:comp_tree}. It is apparent that this does not change the proof.}  We need to compute the probability that any of the four steps in the construction of the neighborhood of a VN introduces a cycle. Note, that only the sets $V^{t,u}$ are random since the MAC connections are fixed. No cycle is introduced if $V^{t,1} \cap N_{1}^{2t} = V^{t,2} \cap N_{2}^{2t+1} = \emptyset$.
	 In addition, since we need to take the same-codebook constraint into account, we also require that 
	 \begin{equation}
	 	V^{t,u} \cap \{N_u^{2t} + \tau\} = \emptyset \quad \forall \tau \in [1:\tau_\text{max}]
	 	\label{eq:cycle}
	 \end{equation}
	 for $u=\{1,2\}$.
	   Note that (28) is necessary because otherwise there would be a $\tau$ for which $N_2^{2t+1}$ contains a node which is a mirrored copy of a node in $N_1^{2t+1}$. This implies that the edges connected to it are fixed and cannot be considered random iid anymore.
	 Even though the event \eqref{eq:cycle} does no necessarily result in a cycle, we treat it as such to get an upper bound on the probability that a computation graph is cycle-free. This increases the number of VNs that results in a cycle by a factor ($1+\tau_\text{max}$) in each iteration compared to the case without MAC connections. Intuitively it is clear that this does not change the basic proof idea of \cite{Ric2001a} since the size of a neighborhood after $t$ iterations does still not scale with $n$. Nonetheless, we give a formal proof for completeness. 
	 Let $c^T_{u}$ and $v^T_{u} = |N^{2T}_{u}|$ denote the number of CNs and VNs in the computation graph of user $u$ after $T$ iterations in $\mathcal{V}_{u}$.  Then, at iteration $t+1$, the number of newly added CNs is at most 
	\begin{equation}
	c^{t+1}_{u} - c^t_{u} \leq v_{u}^t l_\text{max}	
	\label{eq:c_bound}
	\end{equation}
	and the number of newly added VNs is 
	at most 
	\begin{equation}
	v^{t+1}_{u} - v^{t}_{u} \leq c^{t+1}_{u}r_\text{max}.
	\label{eq:v_bound}
	\end{equation}
	Both of these quantities can be upper-bounded independently of the index $u=\{1,2\}$, so we drop it. Furthermore, 
	$v^{T}\geq v^{t}$ and $c^{T}\geq c^{t}$ for $T\geq t$. 
	Conditioned on the event that $\mathcal{N}^{2(T-1)}(\tau)$ is a tree for all $\tau \in [1:\tau_\text{max}]$, going one step deeper will result in no cycles if the edges from the new VNs, of which there are $v^T-v^{T-1}$, meet two conditions: First, they connect to distinct, not yet visited, CNs. And second, the resulting new CNs connect to distinct VNs that are neither in the set of $c^{T-1}$ already visited VNs nor in the same set shifted by some $\tau$. Both copies of the graph follow the same rules and have distinct sets of VNs and CNs, so we can bound them in the same way. The resulting probability is
	\begin{equation}
	\begin{split}
		&\PP(\mathcal{N}^{2T}(\tau) \text{ is a tree } \forall \tau|\mathcal{N}^{2(T-1)}(\tau) \text{ is a tree }\forall \tau) \\
		& \geq
		 \left(1-\frac{(1+\tau_\text{max})c^T}{m}\right)^{(1+\tau_\text{max})(c^T-c^{T-1})}\\
   &\quad\cdot\left(1-\frac{(1+\tau_\text{max})v^T}{n}\right)^{(1+\tau_\text{max})(v^T-v^{T-1})}   
	\end{split}
	\end{equation}
	 So we obtain recursively that 
	\begin{equation}
	\begin{split}	
		&\PP(\mathcal{N}^{2T} \text{ is a tree }\forall \tau) \\
		&\geq  \prod_{t=1}^T \PP(\mathcal{N}^{2t} \text{ is a tree }\forall \tau|\mathcal{N}^{2(t-1)} \text{ is a tree }\forall \tau)\\
		&\geq \left(1-\frac{(1+\tau_\text{max})c^{T}}{m}\right)^{(1+\tau_\text{max})c^{T}}\\
  &\quad\cdot\left(1-\frac{(1+\tau_\text{max})v^{T}}{n}\right)^{(1+\tau_\text{max})v^{T}} \\
		&\geq 1 - \frac{((1+\tau_\text{max})v^T)^2 + \frac{((1+\tau_\text{max})c^T)^2}{1-R}}{n}
	\end{split}
	\end{equation}
	and therefore
	\begin{equation}
		\PP(\mathcal{N}^{2T}_{i} \text{ is not a tree}) \leq \frac{((1+\tau_\text{max})v^T)^2 + \frac{((1+\tau_\text{max})c^T)^2}{1-R}}{n}	
	\end{equation}
	We conclude the proof by giving bounds on $c^T$ and $v^T$
	\begin{align}
	c^T &\leq l_\text{max}\sum_{t=1}^{T-1}v^t \leq l_\text{max}(T-1)v^{T-1}  \\
	    v^T &\leq r_\text{max}Tc^T\leq l_\text{max}r_\text{max}T^2v^{T-1} 
	\end{align}
	which gives
	\begin{align}
		v^T &\leq (l_\text{max}r_\text{max}T^2)^T\\
		c^T &\leq l_\text{max}(T-1)(l_\text{max}r_\text{max}T^2)^{T-1}.
	\end{align}
	We conclude the proof by noting that both upper bounds on $c^T$ and $v^T$ are independent of $n$.
\end{proof}

To conclude the proof of Theorem \ref{thm:de} it remains to show that $\EE_\dv[P_b(\dv)]$ converges to the ensemble average over $\mathcal{G}\in$ LDPC($\lambda,\rho$) as $n\to \infty$. We omit a full proof and give only an outline since it follows, almost without modifications, the proof in \cite{Ric2001a}. By Lemma~\ref{thm:indep} assume that the channel output is iid in the computation of $\EE_{\mathcal{G},\dv}[P_b(\dv)]$. By Lemma~\ref{thm:treelike} we can reduce the computation of $\EE_{\mathcal{G},\dv}[P_b(\dv)]$ to $\EE_{\mathcal{G}}[\mathds{1}(v^l_i==\epsilon)|\mathcal{N}_i^l \text{is a tree}]$ where $v_i^l, i=1,...,2n$, denotes the value of the $i$-th VN after $l$ iterations. The convergence of the edge erasure probabilities to the ensemble average can be shown by constructing an edge exposure martingale. In our case each revealed edge affects both users' graphs so the number of edges affected in the depth $l$ neighborhood doubles. It is apparent that the martingale still has bounded increments as the number of edges in the depth $l$ neighborhood of a given edge does not scale with $n$.   
Together with Lemma \ref{thm:dither_conc} and Corollary \ref{cor:cw_conc} this concludes the proof.
\section{Details on Degree Optimization}
\label{appendix:opt}
The optimization of $\lambda$ for fixed $\rho$ can be expressed in standard form as follows.
 \begin{equation}
 \begin{split} 
 	g_\lambda(x)
 	&= x - \frac{1}{2}L(z(x))\lambda(z(x))\\
 	&= x - \frac{1}{2(\sum \frac{\lambda_i}{i})}\lambdav^T \Hm_x \lambdav
 \end{split}
 \end{equation}
where $H_{x,ij} = \frac{z(x)^{ij-1}}{i}$ and $z(x) = 1 - \rho(1-x)$. 
We get the optimization problem:
\begin{equation}
\begin{split}
\max_{\kappa,\lambda} \quad & \kappa \\
\textrm{s.t.} \quad & \sum \frac{\lambda_i}{i} - \kappa = 0; \lambda_i\geq 0; \sum\lambda_i = 1;\\
 & \lambdav^T \Hm_x \lambdav - 2\kappa(x-\delta)<0\ \forall x \in (0,1) 
\end{split}
\end{equation} 
\section{Error Floor Analysis}
\label{appendix:error_floor}
Throughout this section we use the term $4K$ stopping set ($4K$-SS) to denote stopping sets of size $4K$ consisting of just degree one VNs.
\begin{theorem}
\label{th:4-cycles}
	The probability that a random code from the ensemble LDPC($\lambda,\rho$) results in a joint graph that has no 4-SS for all $\tau \in [1:\tau_\text{max}]$
	can be bounded as
	\begin{multline}     	\PP\left(\mathcal{G}(\tau)\text{ has no 4-SS }\forall \tau \in [1:\tau_\text{max}]\right)\\
      \geq 1 - \tau_\text{max}\frac{L_1^4}{2(1-R)^2}
	\end{multline}
 \hfill $\square$
\end{theorem}
\begin{proof}
	There are ${L_1n \choose 2} \leq n^2L_1^2/2$ pairs of degree one VNs. Let $n_c = (1-R)n$ denote the number of CNs. The probability that a pair of VNs is connected to the same CN is $1/n_c$. Also let $\tilde{p}$ denote the probability that 4-stopping set appears that contains a given pair of VNs $(v_1,v_2)$ in the joint graph with fixed $\tau$. It is given by $\tilde{p} = L_1^2/n_c$, i.e., the probability that the nodes connected to $(v_1,v_2)$ through MAC nodes are both of degree one and connect to the same CN. The degrees and edges of all $\tau_\text{max}$ VNs to the right of $(v_1,v_2)$ are independent and therefore the probability that at least one of the joint graphs with shift $\tau$ contains a 4-SS is given by $1 - (1-\tilde{p})^{\tau_\text{max}}$. Let $N_4$ denote the number of 4-SSs and $I_p$ the event that a 4-SS goes through pair $(v_1,v_2)$. Then the expected number of 4-SSs is given by
	\begin{equation}
	\begin{split}
		\EE[N_4] &= \EE\left[\sum_{p=1}^{L_1n \choose 2} I_p\right] \\
		 &\leq \frac{L_1^2n_c^2}{2(1-R)^2}\frac{1}{n_c}\left(1-\left(1-\frac{L_1^2}{n_c}\right)^{\tau_\text{max}}\right)\\
		 &\leq \tau_\text{max}\frac{L_1^4}{2(1-R)^2} 
	\end{split}
	\end{equation}
	The last inequality follows because $(1-x)^\tau \geq 1 - \tau x$. If the expected number of 4-SSs is smaller than 1 there must be graphs in the ensemble that result in zero 4-SSs. Furthermore, for any non-negative random variable $N$ it holds that $\PP(N=0)\geq 1 - \EE[N]$.  	
\end{proof}
\begin{proof}[Proof of Thm. \ref{thm:4K-cycle}]
Let $k\leq K$.
There are ${L_1n \choose 2k} \leq n^{2k}L_1^{2k}/(2k)!$ $k$-tuples of degree one VNs. Let $n_c = (1-R)n$ denote the number of CNs. For each $2k$-tuple there are $(2k-1)!!$ ways to partition them in pairs, where $(2k-1)!! = (2k-1)(2k-3)...\cdot 1$ denotes the double factorial. The probability that each pair is connected to the same CN is $n_c^{-k}$. Let $\tilde{p}_{2k}$ denote the probability that a $4k$-SS goes through a given $2k$-tuple in the joint graph with fixed $\tau$.
Note that the degrees and edges of the neighbor sequence are not independent if the original tuple of VNs contains a consecutive sequence of length at least three. We show later that their contribution to the expected number of $4k$-SSs is at most of order $\mathcal{O}(1/n^2)$ and results in the correction term in \eqref{eq:4K-cycles}.
For now we consider only $2k$-tuples which do not contain consecutive sequences. For those, the degrees and edges of the neighbor sequence are independent of the original tuple. 
 A $4k$-SS is created if the $2k$-tuple connected by MAC nodes consist of only degree one VNs which connect to $k$ CNs in a configuration that does not result in shorter SSs. With respect to random permutations of VNs and edges this happens with probability
\begin{equation}
 \tilde{p}_{2k} \leq \frac{(2k-1)!!L_1^{2k}}{n_c^k}     
 \end{equation}
 Here we have trivially lower bound the configurations that result in SSs smaller than $2k$ by zero. 
 The probability that at least one of the joint graphs with shift $\tau$ contains a $2k$-SS is given by $1 - (1-\tilde{p}_{2k})^{\tau_\text{max}}$. Let $N_{4k}$ denote the number of $4k$-SSs and $I_{p,2k}$ the event that a $4k$-SS goes through the $2k$-tuple $p$. Then the expected number of $4k$-SSs is given by
	\begin{equation}
	\begin{split}
		\EE[N_{4k}] &= \EE\left[\sum_{p=1}^{L_1n \choose 2k} I_{p,2k}\right] \\
		 &\leq \frac{(2k-1)!!L_1^{2k}n_c^{2k}}{(2k)!(1-R)^{2k}}\frac{1}{n_c^{k}}\left(1-\left(1- \tilde{p}_{2k}\right)^{\tau_\text{max}}\right)\\
   &+ \mathcal{O}\left(\frac{1}{n^3}\right)\\
		  &\leq \tau_\text{max}\left(\frac{L_1^2}{1-R}\right)^k \frac{((2k-1)!!)^2}{(2k)!}+ \mathcal{O}\left(\frac{1}{n^3}\right)\\
            &\leq \tau_\text{max}\left(\frac{L_1^2}{1-R}\right)^k \frac{1}{2k}+ \mathcal{O}\left(\frac{1}{n^2}\right)
	\end{split}
	\end{equation}
	The second inequality follows because $(1-x)^\tau \geq 1 - \tau x$.
 The expected number of SSs up to length $4K$ is $\EE[N_{\leq 4K}] = \sum_{k=1}^{K}\EE[N_{4k}]$ . If $\EE[N_{\leq 4K}]$ is smaller than 1 there must be graphs in the ensemble that result in zero SSs of size smaller than $4K$ because for any non-negative random variable $N$ it holds that $\PP(N=0)\geq 1 - \EE[N]$.  	

 It remains to show that the number of $2k$-tuples that contain consecutive sequences is of order $\mathcal{O}(1/n^2)$.
 The number of length $2l+1$ sequences is at most linear in $n$ while it reduces the probability that the neighbor sequence connects to $k$ CNs by at most a factor of $n_c^l$. Therefore, the expected number of $4k$-SSs that go through at least $2l+1$ consecutive VNs can be bound loosely by a $\mathcal{O}(n/n^{2k-l})$ term. Since $l \leq k-1$ the term is maximized for $l=k-1$ and $k=2$ giving the desired result.  
\end{proof}
\begin{IEEEproof}[Proof of Thm. \ref{thm:BLER_expurg}]
The proof follows by noting that the probability of having stopping sets with VNs with degree larger than one connected to the same set of CNs will go to zero as $n\to\infty$. Indeed,  the smallest possible stopping set containing degree two VNs is the one where two degree one VNs connect to the same CN, and two degree two VNs connected to the same pair of CNs. Their expected number can be upper-bounded by ${L_2n \choose 2}L_1^2/n_c^3 = \mathcal{O}(1/n)$ since $n_c$ scales with $n$. Any larger stopping set containing degree two, or higher, VNs will have an even smaller expected number. Thus, as $n\to\infty$, we can have only stopping sets involving degree one VNs, which implies that expurgating the randomly generated graphs that contains these stopping sets guarantees a vanishing BLER as $n$ grows.  
\end{IEEEproof}
\end{document}